\documentclass{article}

\usepackage[english]{babel}
\usepackage{enumerate}
\usepackage{amsthm} 

\usepackage[a4paper,top=3cm,bottom=3cm,left=3cm,right=3cm,marginparwidth=1.75cm]{geometry}

\usepackage[numbers]{natbib}
\usepackage{pifont}

\usepackage{bm}

\usepackage{amsmath}

\usepackage{amsfonts}
\usepackage{bbm}
\usepackage{graphicx}
\usepackage[colorlinks=true, allcolors=blue]{hyperref}
\usepackage{mathtools}
\usepackage{breqn}
\usepackage{float}
\usepackage{xcolor}
\newtheorem{theorem}{Theorem}[section]
\newtheorem{lemma}[theorem]{Lemma}

\newtheorem{proposition}[theorem]{Proposition}

\newtheorem{example}[theorem]{Example}

\numberwithin{equation}{section}

\def\E{{\mathbb{E}}}
\newcommand{\R}{{\mathbb R}}
\newcommand{\Mid}{{\ \Big|\ }}

\definecolor{blue0}{RGB}{0,77,153} 
\definecolor{red0}{RGB}{179,0,77} 
\definecolor{green0}{RGB}{134,219,76} 
\definecolor{gray0}{RGB}{84,97,110}

\usepackage{mathtools}
\usepackage{authblk}
\mathtoolsset{showonlyrefs}
\title{The quintic Ornstein-Uhlenbeck volatility model that jointly calibrates SPX \& VIX smiles}

\author[1]{Eduardo Abi Jaber\thanks{eduardo.abi-jaber@polytechnique.edu. The first author is grateful for the financial support from the Chaires FiME-FDD, Financial Risks, Deep Finance \& Statistics and Machine Learning and systematic methods in finance at Ecole Polytechnique.}}
\author[2]{Camille Illand}
\author[3]{Shaun (Xiaoyuan)  Li\thanks{shaunlinz02@gmail.com. {{The third author is grateful for the finanical support provided by AXA Investment Managers and} we would like to together thank Salmane Lahdachi at AXA Investment Managers for very fruitful discussions and insightful comments.}}}
\affil[1]{Ecole Polytechnique, CMAP}
\affil[2,3]{AXA Investment Managers}
\affil[3]{Université Paris 1 Panthéon-Sorbonne, CES}

\begin{document}

\maketitle

\begin{abstract}
The quintic Ornstein-Uhlenbeck volatility model  is a stochastic  volatility model where the volatility process is a polynomial function of degree five of a single Ornstein-Uhlenbeck process with fast mean reversion and large vol-of-vol. The model  is able to achieve remarkable joint fits of the SPX-VIX smiles with only 6 effective parameters and an input curve that allows to match certain term structures. We provide several practical specifications of the input curve, study their impact on the joint calibration problem and consider additionally time-dependent parameters to help achieve better fits for longer maturities going beyond 1 year. Even better, the model remains very simple and tractable for pricing and calibration: the VIX squared  is again polynomial in the Ornstein-Uhlenbeck process, leading to efficient VIX derivative pricing by a simple integration against a Gaussian density; simulation of the volatility process is exact; and pricing SPX products {derivatives} can be done efficiently and accurately by standard Monte Carlo techniques with suitable  antithetic and control variates.
\end{abstract}

\begin{description}
\item[JEL Classification:] G13, C63, G10.  
\item[Keywords:] SPX and VIX modeling, Stochastic volatility, Pricing, Calibration.
\end{description}

\section{Introduction}
Since the financial crisis of 2008, derivatives on volatility  became increasingly popular for hedging purposes and   for directional trading especially when combined with the  underlying stock index.   On the US market, the VIX index introduced by the CBOE became one of the most widely followed volatility index.  By construction, the VIX expresses an interpolation between several points of the SPX implied volatility term structure.   This motivates the need for a consistent modeling of the SPX and VIX. 

By joint SPX--VIX calibration problem, we mean the calibration of a model across several maturities to European call/put options on SPX and VIX together with VIX futures.  Such joint calibration turns out to be quite challenging for several reasons: multitude of instruments to be calibrated (SPX/VIX options and VIX futures, so three types of derivatives) across several maturities (to stay consistent with the construction of the VIX), characterized by an upward sloping VIX implied volatility in contrast with the important at-the-money (ATM) SPX skew that becomes more pronounced for smaller maturities.

Several attempts of joint calibration have been made with varying degree of success. However, in general the models and/or the techniques  considered are sophisticated and make use of  jump processes \cite{baldeaux2014consistent,cont2013consistent,kokholm2015joint,pacati2018smiling,papanicolaou2014regime}; non-Markovian rough volatility \cite{bondi2022rough,gatheral2020quadratic,rosenbaum2021deep} and  path-dependent volatility \cite{guyon2022volatility}; multiple-factors  \cite{fouque2018heston,goutte2017regime, guyon2022volatility,romer2022empirical}; optimal transport \cite{guo2022joint, guyon2020joint, guyon2022dispersion}, randomization of the parameters \cite{grzelak2022randomization} and neural SDEs \cite{guyonneural} to just name a few. The proposed solutions for the joint calibration problem are hence rather challenging to put into practice and need specific advanced numerical methods. This is the main motivation of our work.

Recently, the work of  \cite{abi2022joint} identified for the first time a conventional one-factor Markovian continuous stochastic volatility model  that is capable of achieving remarkable fits for a wide range of maturity slices of SPX and VIX implied volatilities together with the term structure of VIX futures.  It is also  shown 
on an extensive empirical study between 2012 and 2022, that contrary to common beliefs, the one factor Markovian model can jointly calibrate SPX and VIX without appealing to multiple-factors, jumps, rough volatility or path-dependency and can achieve even better performances.  In \cite{abi2022joint}, pricing of VIX and SPX derivatives has been done using quantization techniques and neural networks in order to ensure a fair comparison between Markovian and non-Markovian models. 

In the present work, we focus on the Markovian model identified in \cite{abi2022joint} and we  show that the model is tractable in addition to being remarkably flexible. The dynamics of the stochastic volatility process  in this model are given by a polynomial function of degree five  of a single Ornstein-Uhlenbeck process  with fast mean reversion and large vol-of-vol. Hence the name: quintic Ornstein-Uhlenbeck volatility model. The model has only 6 effective parameters  and an input curve that allows to match certain term structures. In particular, we will highlight the role of the input curve on the joint calibration problem: a parametric forward variance curve can be used when calibrating few slices; if the number of slices is increased then the input curve is first extracted from  the forward variance curve of the market and then tweaked in the calibration process. We also consider additionally time-dependent parameters to help achieve better fits for longer maturities going beyond 1 year.

We show that the model is  tractable as it offers an 
explicit expression for the  VIX squared which is again  polynomial in the driving Ornstein-Uhlenbeck factor, leading to efficient VIX derivative pricing by integrating directly against a Gaussian density. Simulation of the volatility process  is exact so that  pricing SPX products can be done efficiently and accurately by standard Monte Carlo techniques with suitable  antithetic and control variates. We also provide a   notebook with our implementation 
 here: \url{https://colab.research.google.com/drive/14nh9civ_wgQv283eshBWnr146w7Xsbi5?usp=sharing}.

For the first time in the literature, remarkable joint fits of SPX and VIX volatility surfaces and VIX futures are achieved between 1 week and beyond 1 year. Although it is challenging, but possible, for another model to achieve similar fits,  it would be very difficult to do so with   a simpler continuous model  than our quintic Ornstein-Uhlenbeck volatility model.

\textbf{Outline.} Section~\ref{S:model} introduces the model. Sections~\ref{S:vix} and \ref{S:spx} detail the pricing of VIX and SPX derivatives in the model. Calibration results on market data are shown in Sections~\ref{S:calibration} and \ref{appendix}.   Appendix A proves the martingality of the underlying process of the model.

\section{The quintic Ornstein-Uhlenbeck volatility model}\label{S:model}
The dynamics of the stock price $S$, with no interest nor dividends, is given by
\begin{equation}\label{polynomial_model}
  \begin{aligned}
    \frac{dS_t}{S_t} &= \sigma_tdB_t,\\
    \sigma_t &= \sqrt{\xi_0(t)}\frac{p(X_t)}{\sqrt{\E \left[p(X_t)^2\right]}}, \quad       p(x) =\alpha_0 + \alpha_1 x + \alpha_3 x^3 + \alpha_5 x^5,
\\
      X_t &= \varepsilon^{H-1/2} \int_0^t e^{-(1/2-H)\varepsilon^{-1}(t-s)} dW_s,
  \end{aligned}
  \end{equation}
with $ B=\rho W + \sqrt{1-\rho^2} W^{\perp}$, $(W,W^{\perp})$ a two-dimensional Brownian motion on a risk-neutral filtered probability space $(\Omega, \mathcal F,(\mathcal F_t)_{t\geq 0}, \mathbb Q )$, $\rho \in [-1,1]$, non-negative coefficients $\alpha_0,\alpha_1,\alpha_3,\alpha_5\geq 0$   ($\alpha_2=\alpha_4 = 0$), $\varepsilon>0$, $H\in(-\infty, 1/2]$ and an input curve $\xi_0 \in L^2([0,T],\mathbb R_{{+}})$ for any $T>0$, allowing the model to match certain term-structures observed on the market. For instance, the normalization  $\sqrt{\E \left[p(X_t)^2\right]}$ allows $\xi_0$ to match the market forward variance curve since 
\begin{align}\label{eq:fwdvarcalib}
    \mathbb E\left[ \int_0^t \sigma_s^2 ds \right] = \int_0^t \xi_0(s)ds, \quad t\geq 0.
\end{align}
The process $X$ driving the volatility is an Ornstein-Uhlenbeck process with a fast mean reversion of order $(1/2-H)\varepsilon^{-1}$ and a large vol-of-vol of order $\varepsilon^{H-1/2}$ for small values of $\varepsilon$, that is 
\begin{align}\label{eq:sde}
    dX_t =  - (1/2-H)\varepsilon^{-1} X_t dt + \varepsilon^{H-1/2}dW_t.
\end{align}
Such parametrizations  are reminiscent of   the fast regimes  extensively studied by \citet{fouque2003multiscale}, see also   \cite[Section 3.6]{fouque2000derivatives}, which corresponds to the case $H=0$. 
They can also be linked to  more complex models such as  jump models   \cite{mechkov2015fast,abireconcile} for $H\leq -1/2$;
and rough volatility models \cite{abi2022joint,abireconcile} for which  $H \in (0,1/2)$ would play the role of the Hurst index.  Letting the parameter $H\in(-\infty, 1/2]$ free in our model  introduces more flexibility    and leads to better fits than in the aforementioned models.   Another advantage of such parametrization is to stabilize the calibrated value of $H$ through time as opposed to calibrating directly on mean reversion and vol-of-vol parameters which are less stable through time, see \cite[Figure 3]{abi2022joint}.

Taking $p$ a polynomial of degree five  allows us to reproduce the upward slope of the VIX smile. Restricting the coefficients $\alpha$ to be non-negative (with $\alpha_2=\alpha_4=0$) ensures the sign of the at-the-money skew to be the same as $\rho$,  see \cite{abi2022joint} for more details, as well as ensuring the martingale property of $S$, whenever $\rho \leq 0$ and  $\alpha_5>0$, see Appendix~\ref{A:mart} below.

We fix $\varepsilon=1/52$ to further reduce the parameters, which gives 6 calibratable parameters:
\begin{align}\label{eq:THeta}
\Theta := \{\alpha_0, \alpha_1, \alpha_3, \alpha_5, \rho, H \},
\end{align}
 plus the input curve $\xi_0(\cdot)$. Numerical experiments show no significant adverse impact on the joint calibration quality by narrowing the number  of parameters.

\section{Pricing VIX derivatives}\label{S:vix}

\paragraph{An explicit expression for the VIX.}
One major advantage of our model is an explicit expression of the VIX. {In continuous time, the VIX can be expressed as
  \begin{equation}\label{eq:vixdef}
    \mbox{VIX}_T^2 = -\frac {2}{\Delta} \E \left[\log(S_{T+\Delta}/S_T)\mid \mathcal{F}_T \right] \times 100^2  = \frac{100^2}{\Delta} \int_T^{T+\Delta}\xi_T(u) du,
  \end{equation}
with $\Delta= 30$ days and  $\xi_T(u):=\E \left[\sigma_u^2\mid \mathcal{F}_T\right]$  the forward variance process which can be computed explicitly in our model as follows. First, we fix $T\leq u$ and rewrite $X$ as
  \begin{equation}\label{eq6}
X_u = {X_T e^{-(1/2-H)\varepsilon^{-1}(u-T)}} + {\varepsilon^{H-1/2}\int_T^u e^{-(1/2-H)\varepsilon^{-1}(u-s)}dW_s}=:Z_{T}^{u} + {G_{T}^{u}} , 
  \end{equation}
  then, setting 
 $$g(u)=\mathbb E[p(X_u)^2],$$ 
we have that
  \begin{equation}\label{eq7}
    \xi_T(u) = \E \left[\sigma_u^2\mid \mathcal{F}_T\right] = \frac{\xi_0(u)}{g(u)} \E \left[\left(\sum_{k=0}^5\alpha_k X_u^k\right)^2 \Mid \mathcal{F}_T \right] = \frac{\xi_0(u)}{g(u)} \E \left[\sum_{k=0}^{10}(\alpha *\alpha)_k X_u^k \Mid \mathcal{F}_T\right],
  \end{equation}
where $(\alpha *\alpha)_k = \sum_{j=0}^{k}\alpha_j\alpha_{k-j}$ is the discrete convolution. Using the Binomial expansion, we can further develop the expression for $\xi_T(u)$ in terms of $Z^u$ and $G^u$ to get
\begin{equation}\label{fwd_var_process}
    \xi_T(u)=\frac{\xi_0(u)}{g(u)}\sum_{k=0}^{10}\sum_{i=0}^{k} (\alpha *\alpha)_k {k\choose i} \left(X_Te^{-(1/2-H)\varepsilon^{-1}(u-T)}\right)^{i}\E \left[(G_T^u)^{k-i} \right],
  \end{equation}
 where we used the fact that $Z_T^u$ is $\mathcal F_T$-measurable and that  $G^u_T$  {is independent of $\mathcal F_T$},  with ${k \choose i }= k!/((k-i)!i!)$ the binomial coefficient. Furthermore, $G^u_T$ is a Gaussian random variable with mean $0$ and variance $\frac{\varepsilon^{2H}}{1-2H}(1-e^{-(1-2H)\varepsilon^{-1}(u-T)})$. Recall that for a Gaussian variable $Y \sim  \mathcal{N} \left(0, \sigma_Y^2 \right)$, its moments $\E \left[Y^p \right]$ for $p\in \mathbb{N}$ can be computed explicitly:
\begin{align}\label{eq:momentgaussian}
  \E \left[Y^p \right] =
    \begin{cases}
      0 & \text{if $p$ is odd}\\
      \sigma_Y^{p}(p-1)!! & \text{if $p$ is even}\\
    \end{cases}       
\end{align}
with $p!!$ the double factorial. Therefore all moments of $\E \left[(G_T^u)^i \right]$ are given explicitly.

Going back to \eqref{eq:vixdef} and plugging the expression \eqref{fwd_var_process},  the explicit expression of the $\mbox{VIX}_T^2$ turns out to be polynomial in $X_T$:
\begin{align}
 \mbox{VIX}_T^2 &= \frac{100^2}{\Delta} \sum_{k=0}^{10} \sum_{i=0}^{k} (\alpha *\alpha)_k {k\choose i} \int_T^{T+\Delta} \frac{\xi_0(u)}{g(u)}\mathbb{E}\left[(G_T^u)^{k-i} \right] e^{-(1/2-H)\varepsilon^{-1}(u-T)i}du  X_T^i \\ 
    &= \frac{100^2}{\Delta} \sum_{i=0}^{10} \sum_{k=i}^{10} \left((\alpha *\alpha)_k {k\choose i} \int_T^{T+\Delta} \frac{\xi_0(u)}{g(u)}\mathbb{E}\left[(G_T^u)^{k-i} \right] e^{-(1/2-H)\varepsilon^{-1}(u-T)i}du \right) X_T^i \\
    &= \frac{100^2}{\Delta} \sum_{i=0}^{10}  \beta_i X_T^{i},\label{eq:VIXclosed}\\
\end{align}
where
\[
    \beta_i = \sum_{k=i}^{10} (\alpha *\alpha)_k {k\choose i} \int_T^{T+\Delta} \frac{\xi_0(u)}{g(u)} \E \left[(G_T^u)^{k-i} \right] \left(e^{-(1/2-H)\varepsilon^{-1}(u-T)i}\right) du.
\]
The integral  {inside $\beta_i$} can be easily computed, at least numerically for a variety of choices for $\xi_0(\cdot)$.

\paragraph{Pricing VIX derivatives.} Thanks to the closed expression of \eqref{eq:VIXclosed}, $\mbox{VIX}_T^2$ is a polynomial in $X_T$ that we denote by $h(X_T)$. Since $X_T$ is Gaussian with mean $0$ and variance $\sigma_{X_T}^2 = \frac{\varepsilon^{2H}}{1-2H}(1-e^{-(1-2H)\varepsilon^{-1}T})$, pricing $\mbox{VIX}$ derivatives with payoff function $\Phi$ is immediate by integrating directly against the standard Gaussian density:
\begin{equation}\label{vix_derivative}
\E \left[\Phi(\mbox{VIX}_T) \right] = \E \left[\Phi\left(\sqrt{h(X_T)}\right) \right] = \frac 1 {\sqrt{2\pi}} \int_{\R} \Phi\left(\sqrt{h\left( \sigma_{X_T} x\right)}\right)e^{-x^2/2}{dx}.
\end{equation}

\begin{example}
 To price VIX future prices, set $\Phi(v) = v$ and to price VIX vanilla call price, set $\Phi(v) = (v-K)^{+}$. This integral \eqref{vix_derivative} can be computed efficiently using a variety of quadrature techniques. The Gaussian quadrature with $400$ nodes seems to be more than enough to price accurately VIX call and future prices.
\end{example}

\section{Pricing SPX derivatives}\label{S:spx}

To price SPX derivatives, we resort to using Monte Carlo simulations. Since $X$ is a Ornstein-Uhlenbeck process, it can be simulated exactly as opposed to using the Euler scheme which is often inaccurate in a fast mean reversion regime. To simulate $X$, first define
\[
\Tilde{X}_t = X_t e^{\frac{1/2-H}{\varepsilon}t} = \varepsilon^{H-1/2}\int_0^t e^{\frac{1/2-H}{\varepsilon}s}dW_s.
\]
Then, $\Tilde{X}$ can be simulated recursively by
\[
\Tilde{X}_{t_{i+1}} = \Tilde{X}_{t_{i}}+ \sqrt{\varepsilon^{2H}/(1-2H)}\left(e^{\frac{1-2H}{\varepsilon}t_{i+1}}-e^{\frac{1-2H}{\varepsilon}t_{i}}\right)Y_i,
\]
with $Y_i$ i.i.d. standard Gaussian. To get back to $X_{t_{i+1}}$ we just divide $\Tilde{X}_{t_{i+1}}$ by $e^{\frac{1/2-H}{\varepsilon}t_{i+1}}$. This setting allows us to easily vectorize computations.

To simulate the process $\log(S)$, we use the Euler scheme together with antithetic and control variates, the so called  turbocharging method as outlined in \cite{mccrickerd2018turbocharging}. This means we only need to simulate the part of $\log(S)$ that is $\mathcal{F}^W$ measurable, we call this $S^W$ and can be simulated as
\[
\log(S^W)_{t_{i+1}} = \log(S^W)_{t_{i}}-1/2 \left(\rho \sigma_{t_{i}} \right)^2 \left(t_{i+1}-t_{i}\right)+ \rho \sigma_{t_i} \sqrt{t_{i+1}-t_{i}}Y_i.
\]
The main idea of the turbocharging method is to 1) take advantage of the conditional log-normality of $S$ with respect to   $\mathcal{F}^{W}$, hence removing the MC error from simulating $W^{\perp}$, and 2) apply the control variate in the form of a time option where one can again take advantage of the log-normality and closed form solution. We refer readers to \cite{mccrickerd2018turbocharging}}  for more details on the method and to the   notebook mentioned in the introduction for our implementation.

\section{SPX/VIX Joint calibration}\label{S:calibration}

We now address the SPX-VIX joint calibration problem, that is the  calibration of our model to SPX European options, VIX European options and VIX futures across several maturities. Ideally, one should calibrate for SPX options maturity up to one month ahead of that of the VIX options, given that VIX encodes expected level of volatility for the next 30 days by definition. 

The calibration of VIX futures is necessary as it is used to calculate VIX implied volatility. Recall the implied volatility is calculated by inverting  the Black and Scholes formula, that is, for a given call price $C_0 (K,T)$ with strike $K$ and maturity $T$, we find the unique $\sigma_(K,T)$ such that 
  \begin{equation}
    C_0 (K,T) = F(T)\mathcal{N}(d_1)-K\mathcal{N}(d_2)
  \end{equation}
with
  \begin{equation}
d_1 = \frac{\log\left(F(T)/K\right)+ \frac 1 2  \sigma(K,T)^2 T}{ \sigma(K,T) \sqrt{T}}, \quad d_2 = d_1 - \sigma(K,T) \sqrt{T},
  \end{equation}
where  $\mathcal{N}(x)$ is the cumulative density function of the standard Gaussian distribution and $F(T)$ denotes the futures price of the index: $F(T)=\E\left[S_T\right]=S_0$ for the SPX in our model \eqref{polynomial_model} and  $F(T)=\E\left[\mbox{VIX}_T\right]$ for the VIX.

To calibrate our model, we solve the following optimisation problem involving sum of root mean squared error (RMSE):
\begin{equation}\label{calibration_error}
  \begin{aligned}
    &\min_{\Theta} \Bigg\{  c_1 \sqrt{\sum_{i,j}  \Big( \sigma_{spx}^{\Theta}(T_i,K_j) - \sigma_{spx}^{mkt}(T_i,K_j) \Big)^2} +    c_2 \sqrt{\sum_{i,j} \Big( \sigma_{vix}^{\Theta}(T_i,K_j) - \sigma_{vix}^{mkt}(T_i,K_j) \Big)^2}\\ 
    &\quad \quad \quad \quad \quad  +  c_3 \sqrt{\sum_{i} \Big(F^{\Theta}_{vix}(T_i)-F^{mkt}_{vix}(T_i) \Big)^2}\Bigg\}.
  \end{aligned}
\end{equation}

Here, $\sigma_{spx}^{mkt}(T_i,K_j)$, $\sigma_{vix}^{mkt}(T_i,K_j)$ represent market SPX-VIX implied volatility with maturity $T_i$ and strike $K_j$. $F^{mkt}_{vix}(T_i)$ is the market VIX futures price maturing at $T_i$. $\sigma_{spx}^{\Theta}(T_i,K_j)$, $\sigma_{vix}^{\Theta}(T_i,K_j)$ and $F^{\Theta}_{vix}(T_i)$ represent the same instruments, but coming from our model. The coefficients $c_1$, $c_2$ and $c_3$ are some positive numbers used to assign different weights to the errors in SPX-VIX implied volatility and VIX futures price. We chose arbitrarily $c_1 = 1,c_2 = 0.1, c_3 = 0.5$ for our numerical experiments.

We will now show how our model is well adapted to produce joint fits between SPX/VIX{, with daily SPX/VIX joint implied volatility surface data purchased from the CBOE website \url{https://datashop.cboe.com/}}. 

\paragraph{Extracting the forward variance curve $\xi_0(\cdot)$.}  
Using  the well-known replication formula for the log-contract in \cite{carr2001towards}, we construct $\xi_0(\cdot)$ such that  
\begin{align}\label{carr_madan_formula}
    \int_{T_i}^{T_{i+1}} \xi_0(s)ds &= 2\left(\int_0^{S_0}\frac{P_0(K,T_{i+1})}{K^2}dK + \int_{S_0}^{+\infty}\frac{C_0(K,T_{i+1})}{K^2}dK\right) \\
    &- 2\left(\int_0^{S_0}\frac{P_0(K,T_{i})}{K^2}dK + \int_{S_0}^{+\infty}\frac{C_0(K,T_{i})}{K^2}dK\right),
\end{align}
where $T_i$ are SPX option maturities from market data and $C_0(K,T)$ and $P_0(K,T)$ the price of a call/put option with strike $K$ and maturity $T$. Since market prices for out of the money call/put options are not always available,  we first interpolate the SPX market implied volatility surface, each slice separately, using methods like SVI or SABR (after checking for arbitrage), and use the fitted surface to compute the integral above.

We then approximate $\xi_0(t)$ by passing a cubic spline interpolation with nodes $(t_i, x_i)$ with $t_i=(T_i+T_{i+1})/2$ and $x_i = \sqrt{\int_{T_i}^{T_{i+1}} \xi_0(s)ds}$ and then square the interpolation to ensure positivity of the forward variance curve. Of course, piece-wise constant between $[T_i, T_{i+1})$ can also be used.

During the calibration procedure, we will let the optimisation algorithm move the model parameters $\Theta$ defined in \eqref{eq:THeta} and make adjustments to the value of the spline nodes $x_i$ as necessary to jointly fit SPX and VIX derivatives.

Figure \ref{fig:joint_fit} shows the joint fit on the 23 October 2017, with calibrated parameters $\rho = -0.6843, H = -0.0358, (\alpha_0, \alpha_1, \alpha_3, \alpha_5) = (0.5907,1,0.2893,0.0549)$:

  \begin{figure}[H]
    \centering
    \includegraphics[width=0.5\textwidth]{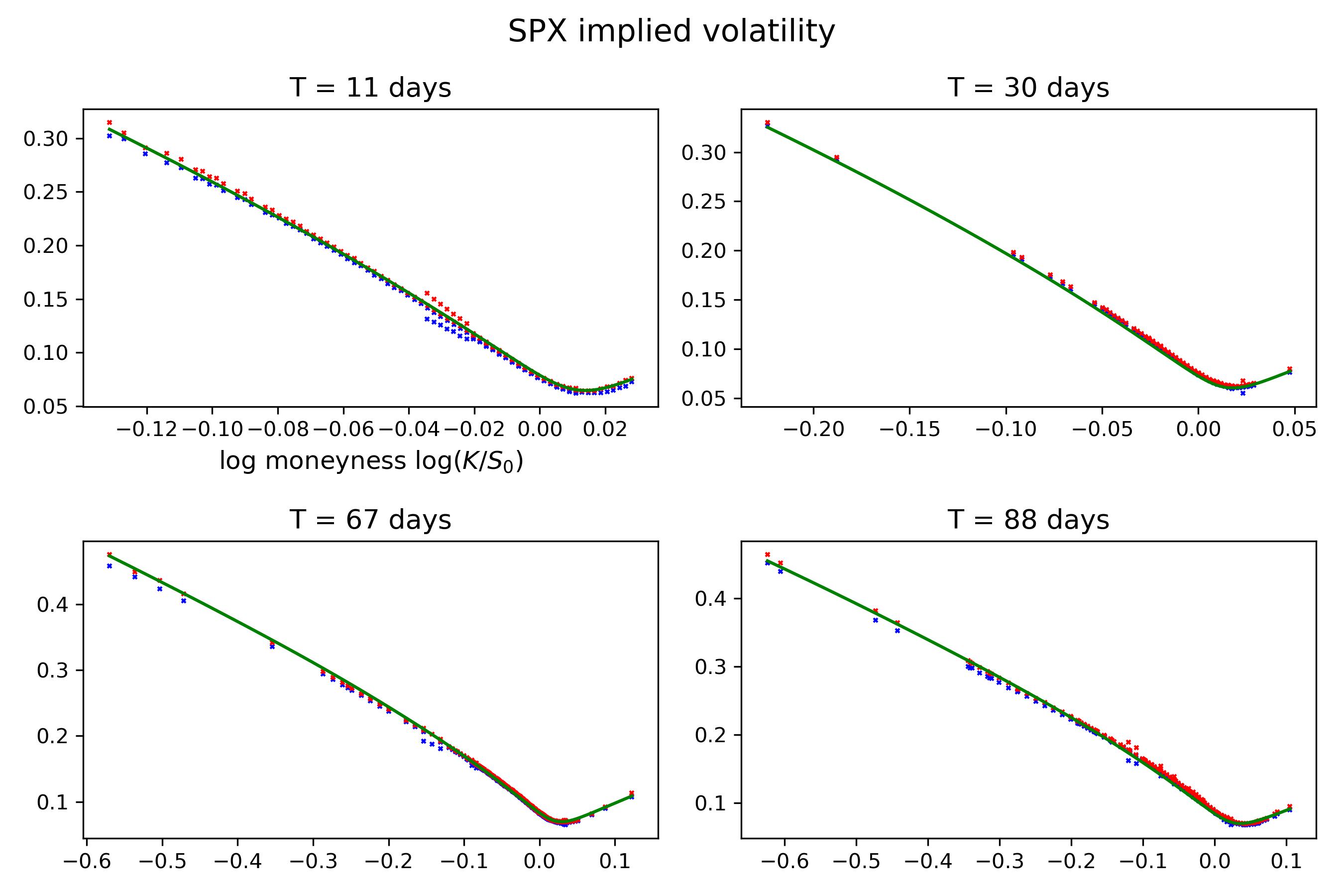}%
    \includegraphics[width=0.5\textwidth]{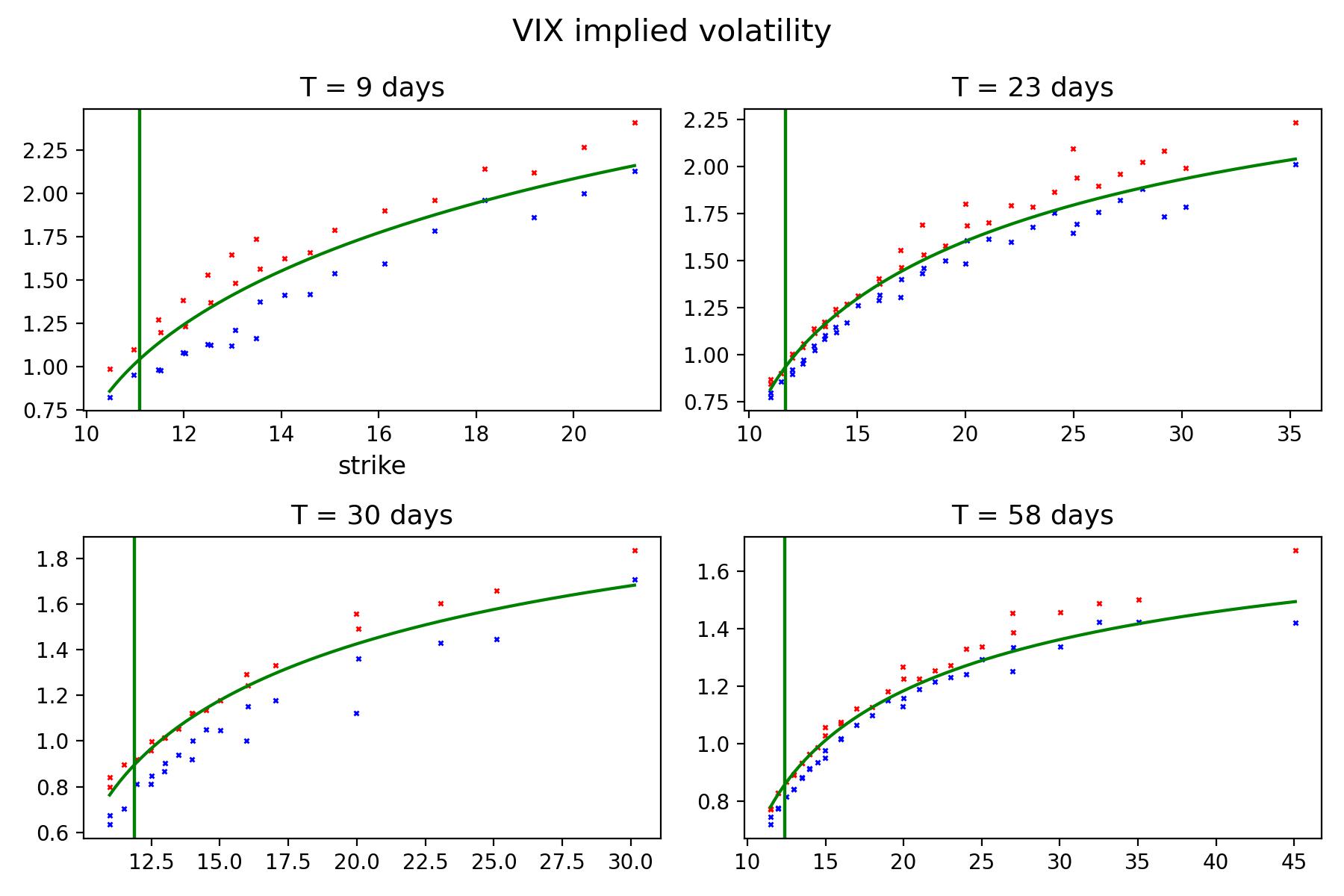}
    \caption{SPX--VIX smiles (bid/ask in blue/red) and VIX futures (vertical black lines) jointly calibrated with our model (full green lines) for 23 October 2017.}
    \label{fig:joint_fit}
  \end{figure}

The forward variance curve has been adjusted to jointly fit the SPX and VIX smiles as shown in Figure \ref{fig:calib_fvc}.

  \begin{figure}[H]
    \centering    \includegraphics[width=0.9\textwidth]{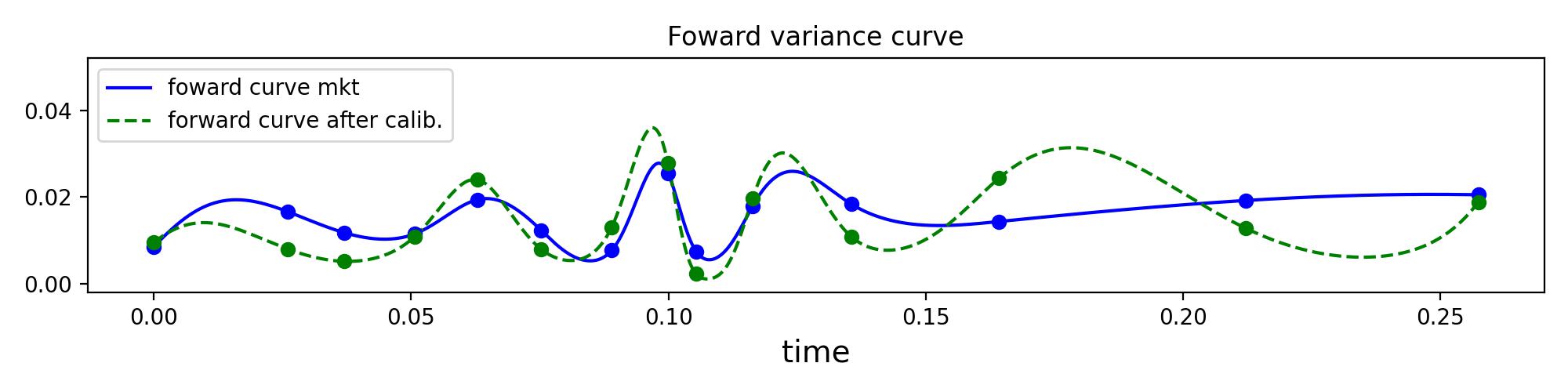}%
    \caption{The blue line represent the forward variance curve stripped from market data using Carr-Madan log contract formula in \eqref{carr_madan_formula}, the dotted green line is the adjusted forward variance curve as part of the calibration to jointly fit SPX/VIX smiles on 23 October 2017, with the round points representing cubic spline nodes.}
    \label{fig:calib_fvc}
  \end{figure}

For more joint surface fits and an empirical study on joint SPX/VIX volatility surface between 2011 and 2022 for our quintic Ornstein-Uhlenbeck model together with its calibrated parameters across time, we refer the reader to \cite{abi2022joint}.

\paragraph{Using parametric forward variance curves {when calibrating fewer slices}.} Instead of extracting forward variance curves from market data, it is also possible to use a parametric form of the forward variance curve for example in the form of:
\begin{align}\label{eq:xi0param}
    \xi_0(t) = a e^{-bt}+c(1-e^{-bt}),
\end{align}
with $a,b,c >0$ to be calibrated. 

The parametric forward variance curves offers less flexibility than that of extracted market forward variance curve discussed before given its rigid form. However, it is still capable to fit two maturity slices of SPX and one slice of VIX. We provide two examples here, with 

\begin{enumerate}
    \item joint fits of SPX options maturing in 9 days and 30 days, and VIX options maturing in 9 days) using parameters $\rho = -0.7316, H = -0.1382, (\alpha_0, \alpha_1, \alpha_3, \alpha_5) = (0.8169, 0.274, 0.1717, 0.0036),$ $ a = 0.0084, b = 2.0436, c = 0.0441$ shown in  Figure \ref{fig:joint_para1},
    \item joint fits of SPX options maturing in 53 days and 88 days, and VIX options maturing in 58 days) with parameters $\rho = -0.7001, H = 0.141, (\alpha_0, \alpha_1, \alpha_3, \alpha_5) = (0.7558, 1, 0.0885, 0.4421),$ $a = 0.012, b = 2.027, c = 0.033$ shown in Figure \ref{fig:joint_para2}.
\end{enumerate}

  \begin{figure}[H]
    \centering
    \includegraphics[width=0.6\textwidth]{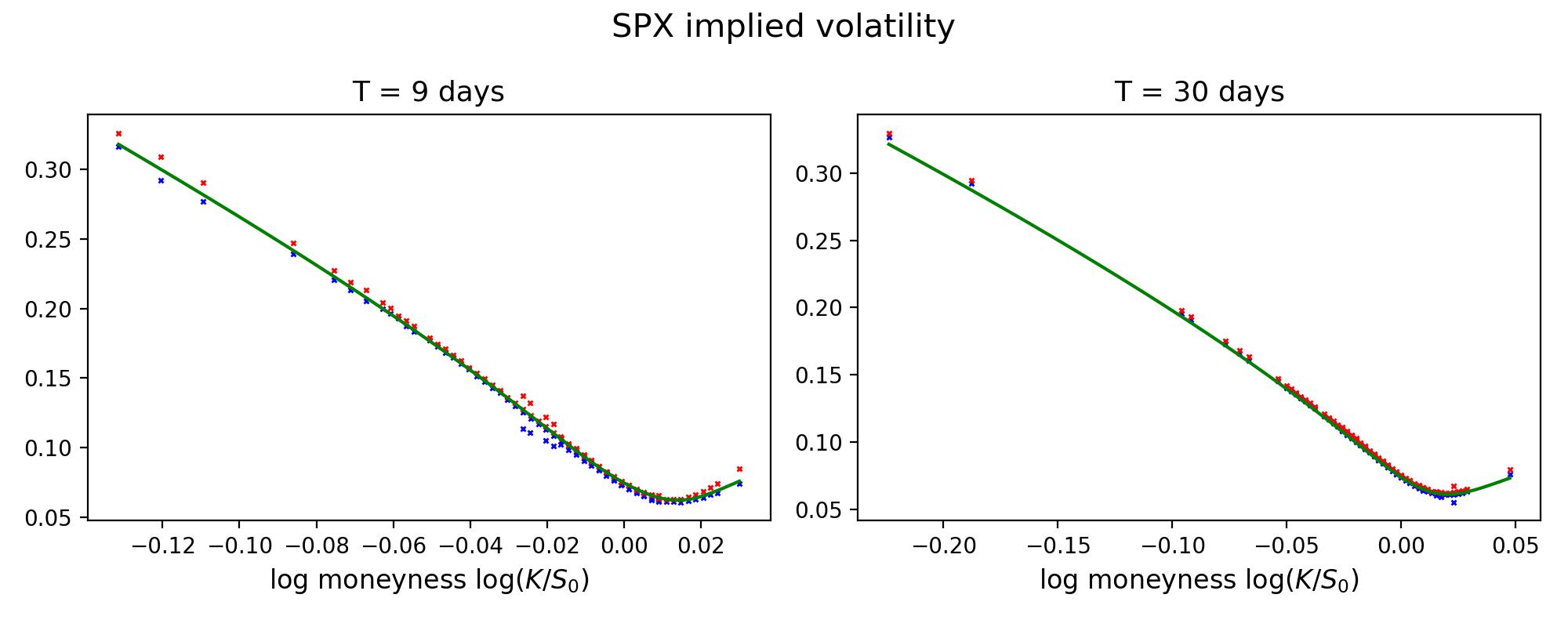}%
    \includegraphics[width=0.3\textwidth]{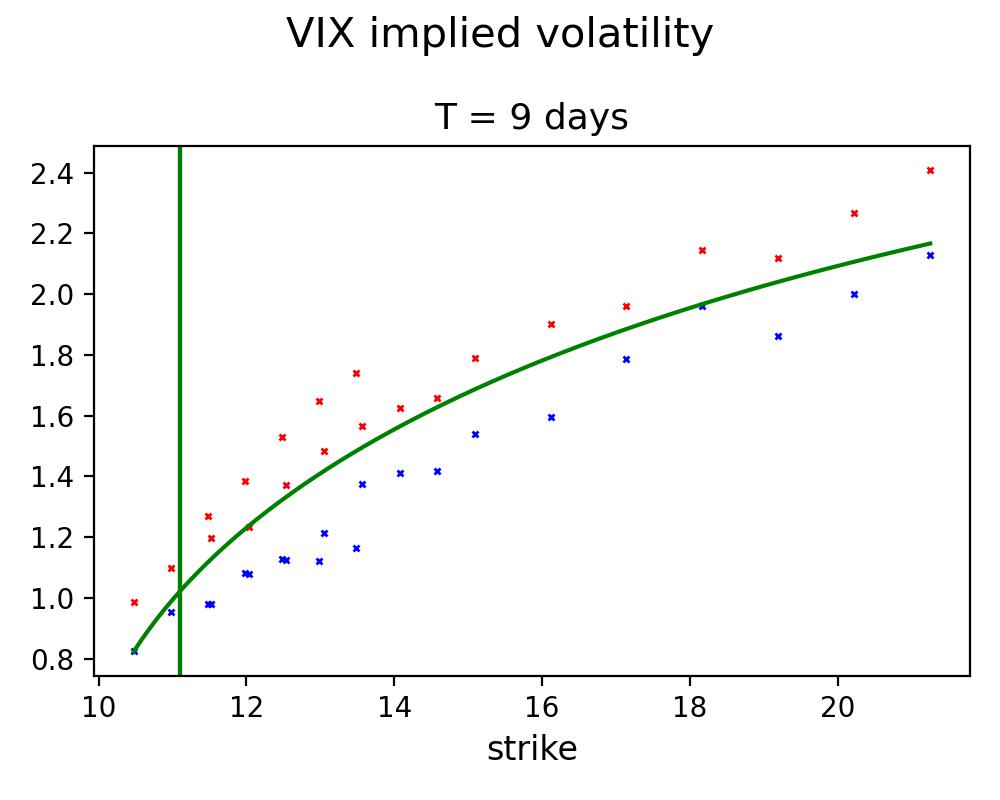}
    \caption{SPX--VIX smiles (bid/ask in blue/red) and VIX futures (vertical black lines) jointly calibrated with our model with  the parametric forward variance curve \eqref{eq:xi0param} for 23 October 2017.} 
    \label{fig:joint_para1}
  \end{figure}

  \begin{figure}[H]
    \centering
    \includegraphics[width=0.6\textwidth]{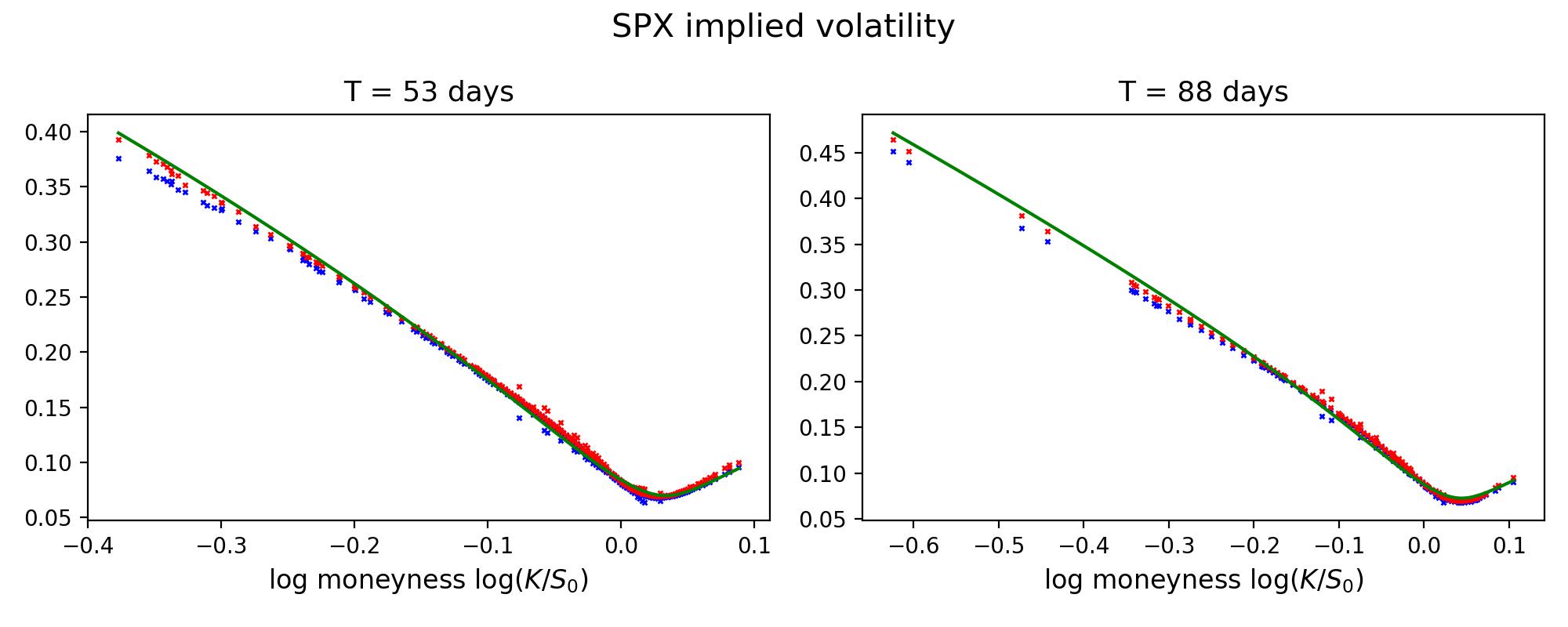}%
    \includegraphics[width=0.3\textwidth]{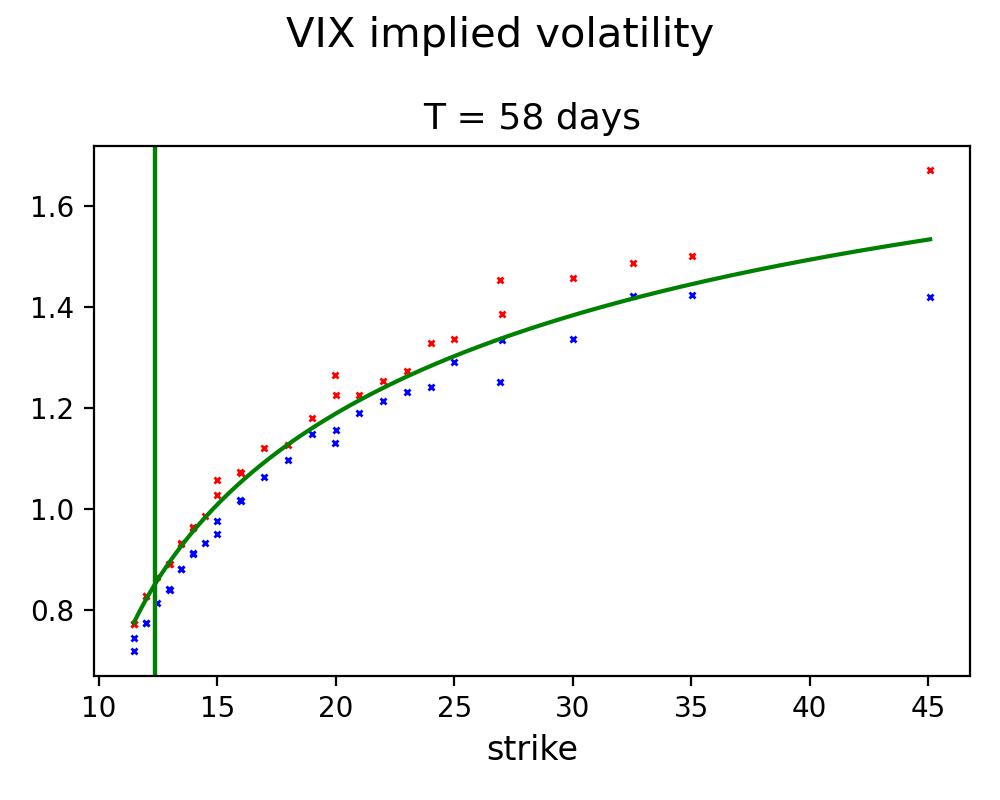}
    \caption{SPX--VIX smiles (bid/ask in blue/red) and VIX futures (vertical black lines) jointly calibrated with our model with  the parametric forward variance curve \eqref{eq:xi0param} for 23 October 2017.}
    \label{fig:joint_para2}
  \end{figure}

\paragraph{Using a time dependent $H$ parameter for fitting  longer maturities.}

To fit even longer maturities beyond 3 and 4 months, we propose to use a time-dependent parametrization of $H$ in \eqref{eq:sde} in the form of
\begin{align}\label{eq:Htime}
H(t) = H_0 e^{-\kappa t} + H_{\infty} (1-e^{-\kappa t}),    
\end{align}
with $H_0, H_{\infty}, \kappa>0$ to be calibrated. With this formulation, $X$ remains a Gaussian Ornstein-Uhlenbeck process with time dependent parameters and can also be simulated exactly. The formula for $\mbox{VIX}_T^2$ remains polynomial in $X_T$ similar to \eqref{eq:VIXclosed}.

Using time dependent parametrization of $H$, together with minor tweaks to the stripped forward variance curve using \eqref{carr_madan_formula} and letting the mean reversion speed $\varepsilon$ free, we can jointly fit the SPX and VIX surface beyond 1 year, with up to 8 slices for SPX and 6 slices for VIX as illustrated on Figure~\ref{Fig:calib_time_dependdent_H}. The calibrated parameters are $\rho = -0.7466, (\alpha_0, \alpha_1, \alpha_3, \alpha_5) = (0,0.0266,0.2513,0.00006)$, $H_0 = 0.3176, H_{\infty} = -1.3665, \kappa = 1.2, \varepsilon = 0.1359$. Figure \ref{fig:fwd_var_curve_time_dependent} shows the forward variance curve $\xi_0(t)$ on 23 October 2017 stripped from the market vs. slightly adjusted forward variance as part of the joint calibration, and Figure \ref{fig:time_dependent_H} shows the value of the calibrated  $H$ in \eqref{eq:Htime} as function of time.

  \begin{figure}[H]
    \centering
    \includegraphics[scale=0.3]{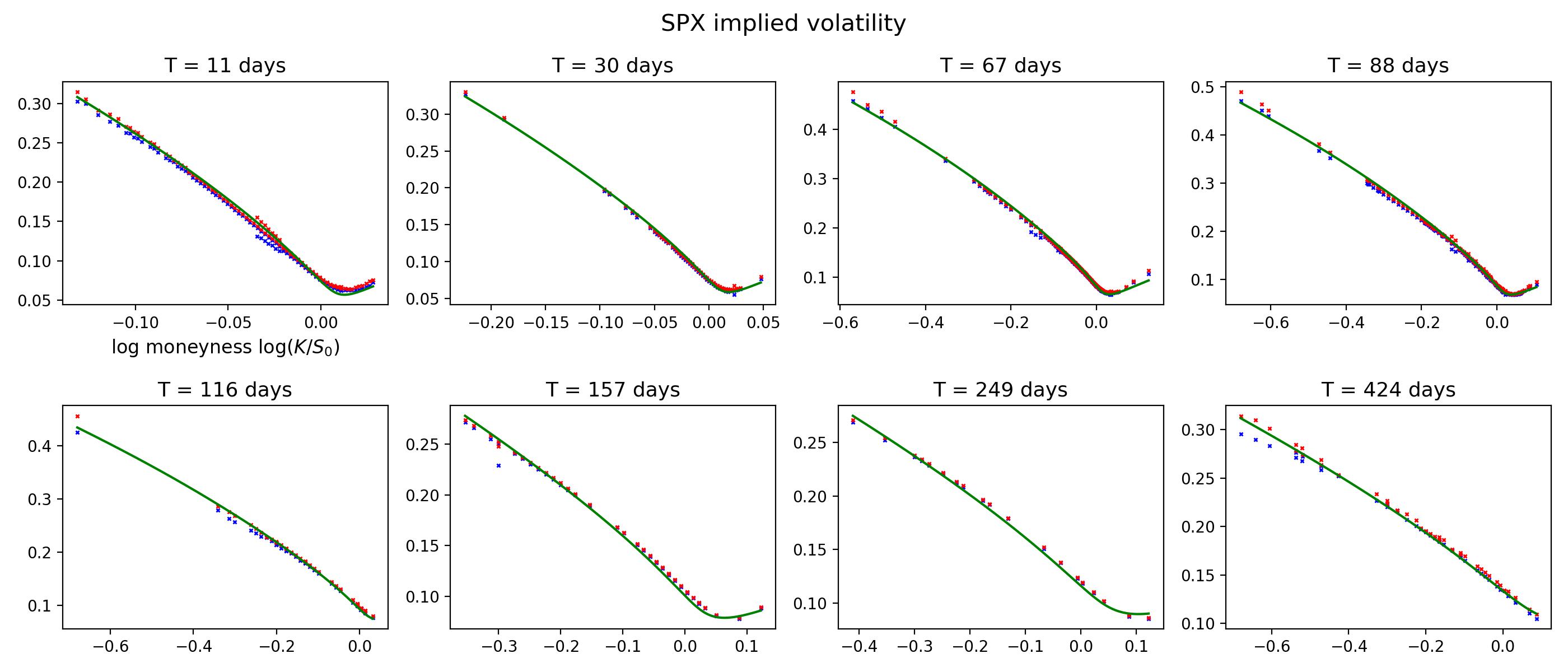}
    \includegraphics[scale=0.3]{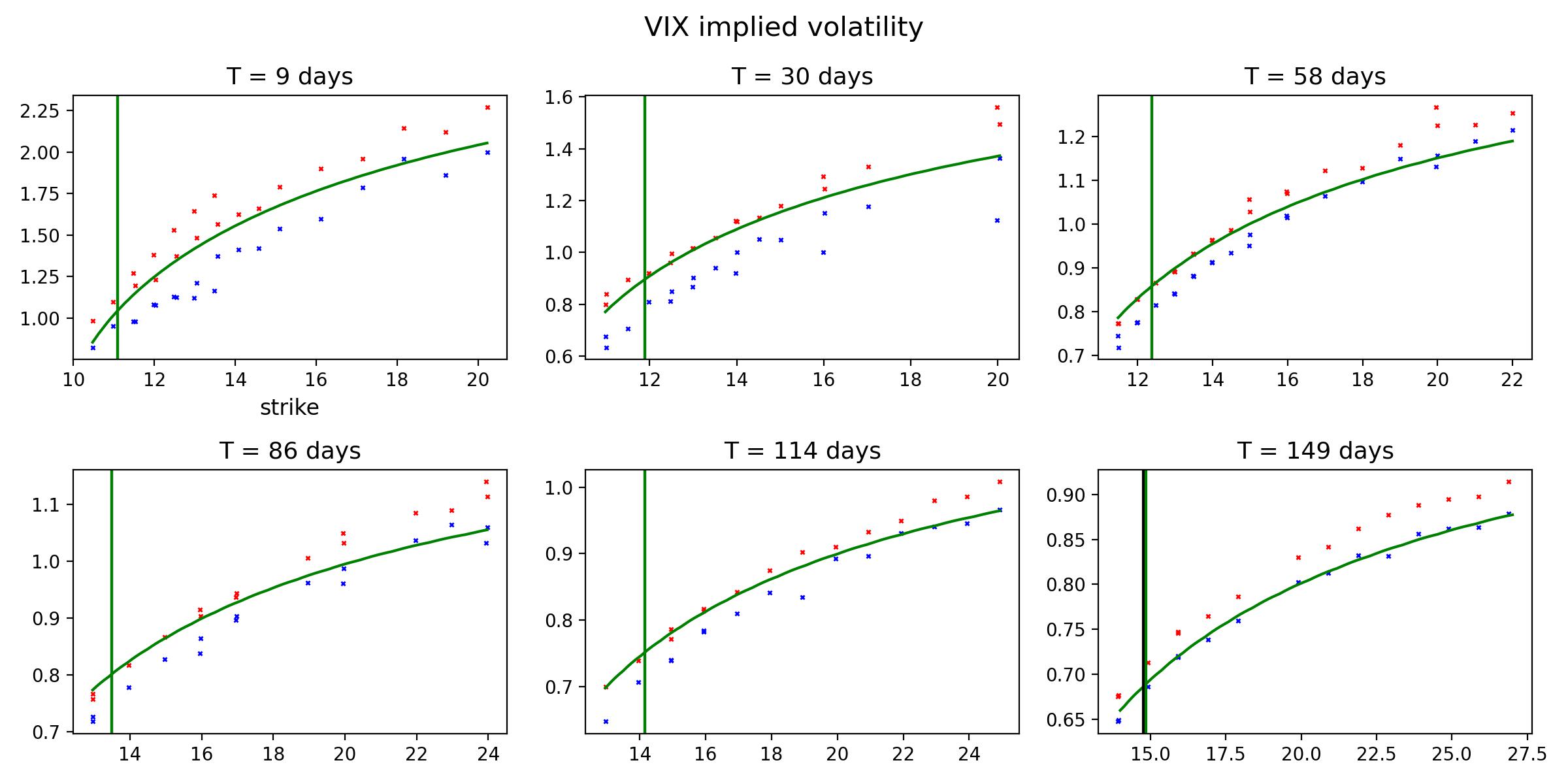}
    \caption{SPX--VIX smiles (bid/ask in blue/red) and VIX futures (vertical black lines) jointly calibrated with our model for time dependent $H$ (full green lines) for 23 October 2017.}
    \label{Fig:calib_time_dependdent_H}
  \end{figure}

  \begin{figure}[H]
    \centering    \includegraphics[width=0.9\textwidth]{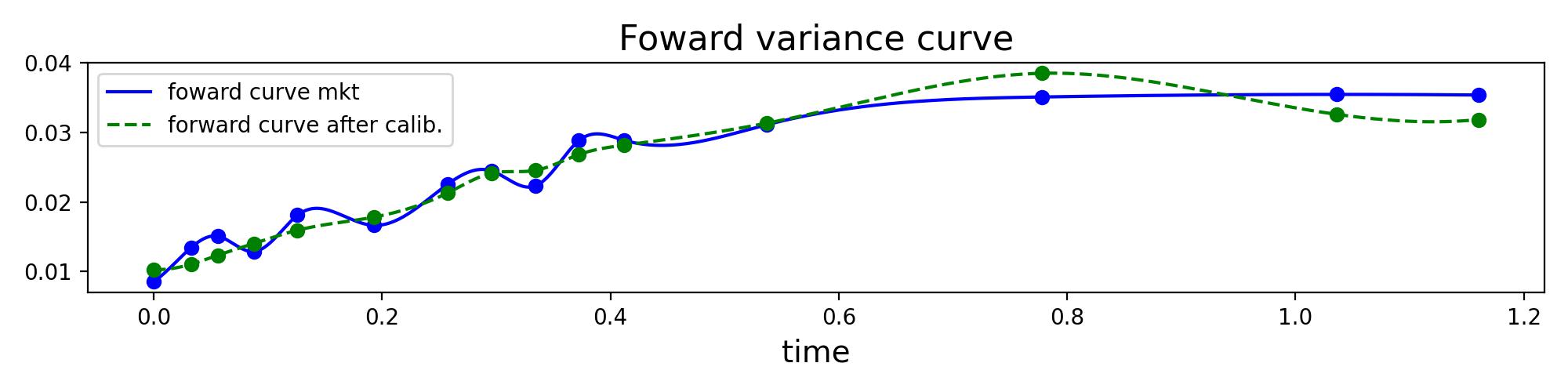}%
    \caption{The blue line represent the forward variance curve stripped from market data using Carr-Madan log contract formula as in \eqref{carr_madan_formula}, the dotted green line is the adjusted forward variance curve as part of the calibration to jointly fit SPX/VIX smiles on 23 October 2017, with the round points representing cubic spline nodes.}
    \label{fig:fwd_var_curve_time_dependent}
  \end{figure}

  \begin{figure}[H]
    \centering    \includegraphics[width=0.9\textwidth]{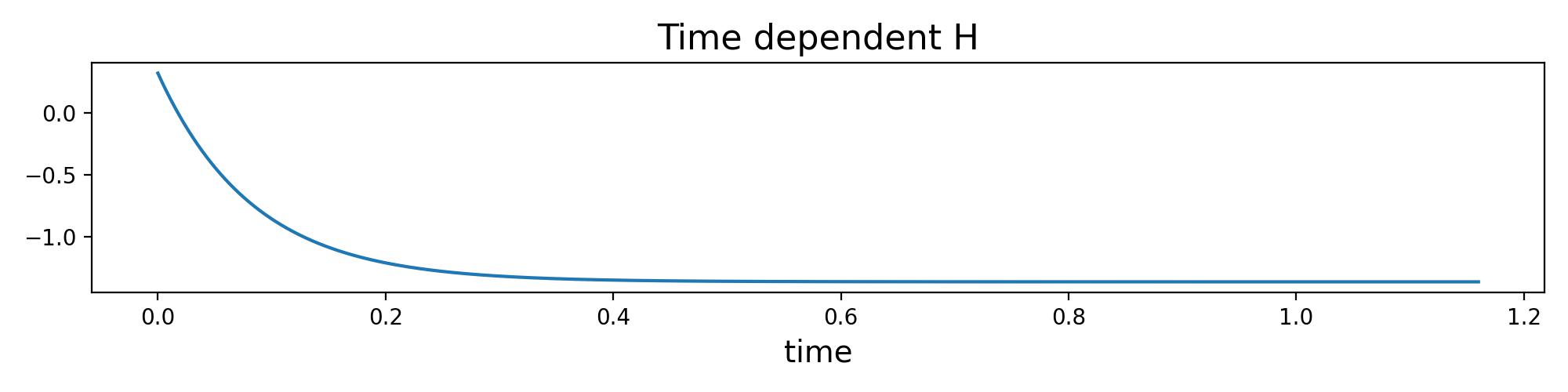}%
    \caption{Value of $H$ as a function of time as part of the calibration to jointly fit SPX/VIX smiles on 23 October 2017.}
    \label{fig:time_dependent_H}
  \end{figure}

\paragraph{Recommendation on the choice of the forward curve $\xi_0$ and the parameters $\varepsilon
$ and $H$ for practical use of the model:}
\begin{itemize}
    \item 
parametric $\xi_0$ as in \eqref{eq:xi0param}, fixed $\varepsilon=1/52$ and constant coefficient $H$ in \eqref{eq:sde} for fits of single slice of VIX and two slices of SPX, 
        \item 
tweaked stripped forward curve $\xi_0$, fixed $\varepsilon=1/52$ and constant $H$ in \eqref{eq:sde} for joint fits on several maturities up to 3 to 4 months,
\item tweaked stripped forward curve $\xi_0$, letting $\varepsilon$ free  and   time-dependent $H$  in \eqref{eq:sde}  as  in \eqref{eq:Htime} for joint fits on several maturities up to 18 months.
\end{itemize}

\section{Additional graphs}\label{appendix}
\subsection{Evolution of calibrated model parameters}\label{calibrated_parameters}
 In this section, we plot the evolution of all calibrated model parameters as part of the joint calibration exercise in \cite{abi2022joint}, where a total of 1,422 days of SPX and VIX joint implied volatility between 2012 and 2022 were calibrated. All model parameters appears to be stable across time, which is desirable from a practical point of view.

  \begin{figure}[H]
    \centering

    \includegraphics[scale=0.3]{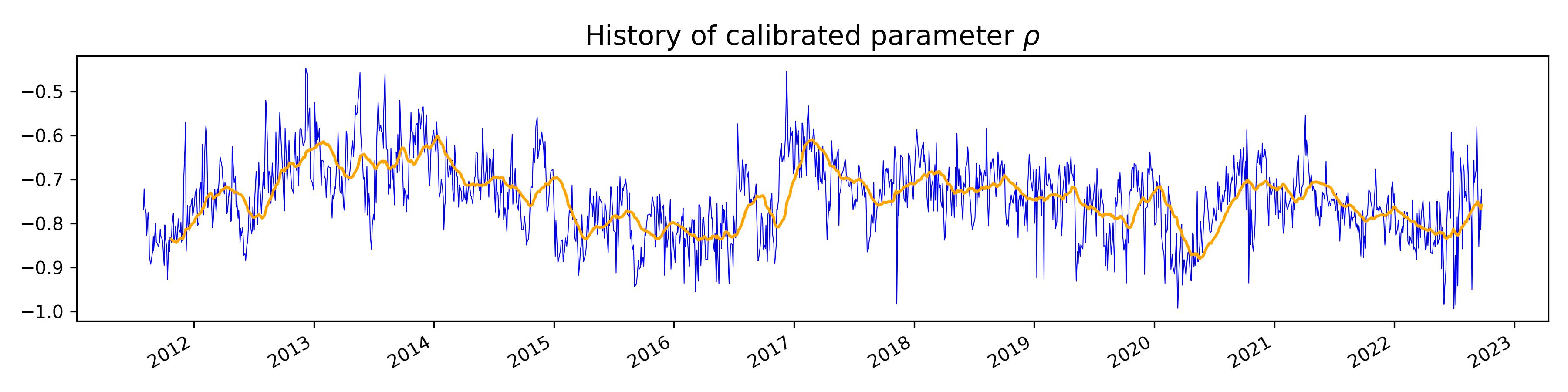}
    \includegraphics[scale=0.3]{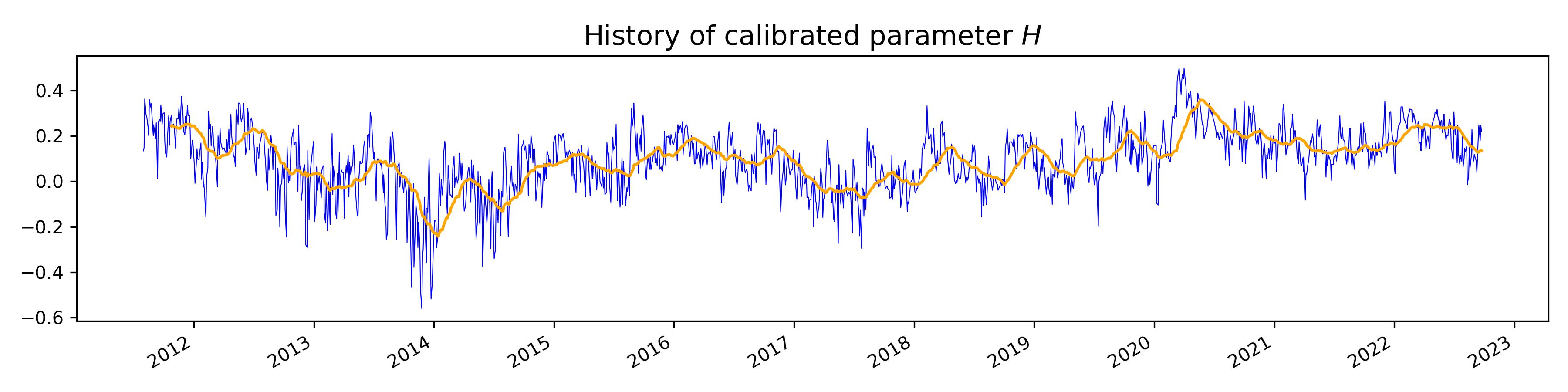}
    \includegraphics[scale=0.3]{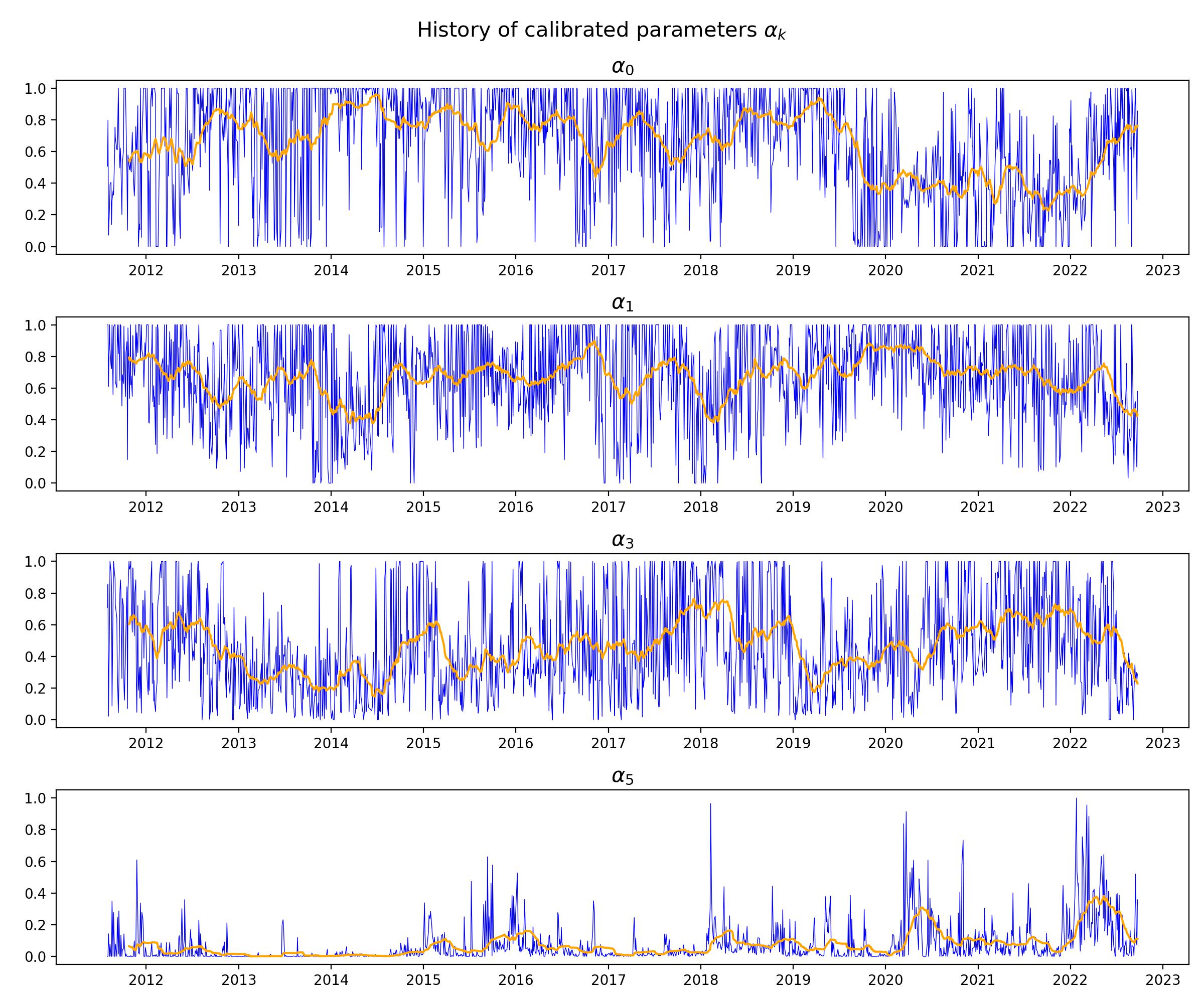}
    \caption{Evolution of the calibrated parameters from the quintic Ornstein Uhlenbeck volatility model. The blue line is the actual value of the calibrated parameters, the orange line is the 30-day moving average.}
    \label{fig:calib_params}
  \end{figure}

\subsection{Model calibration error}\label{calibrated_error} In this section, we take some of the examples provided in the previous sections and re-calibrate the quintic Ornstein Uhlenbeck volatility model to a narrower range of moneyness (near the money).  We then plot the absolute calibration error (model implied volatility vs.~mid implied volatility from market data). To facilitate comparison, we plot the absolute calibration error as a multiplier of half of the bid-ask spread, i.e.~(absolute calibration error)/(0.5 $\times$ bid-ask spread). A multiplier of less than 1 means the model implied volatility is within the bid-ask spread of market data.

  \begin{figure}[H]
    \centering
    \includegraphics[width=0.5\textwidth]{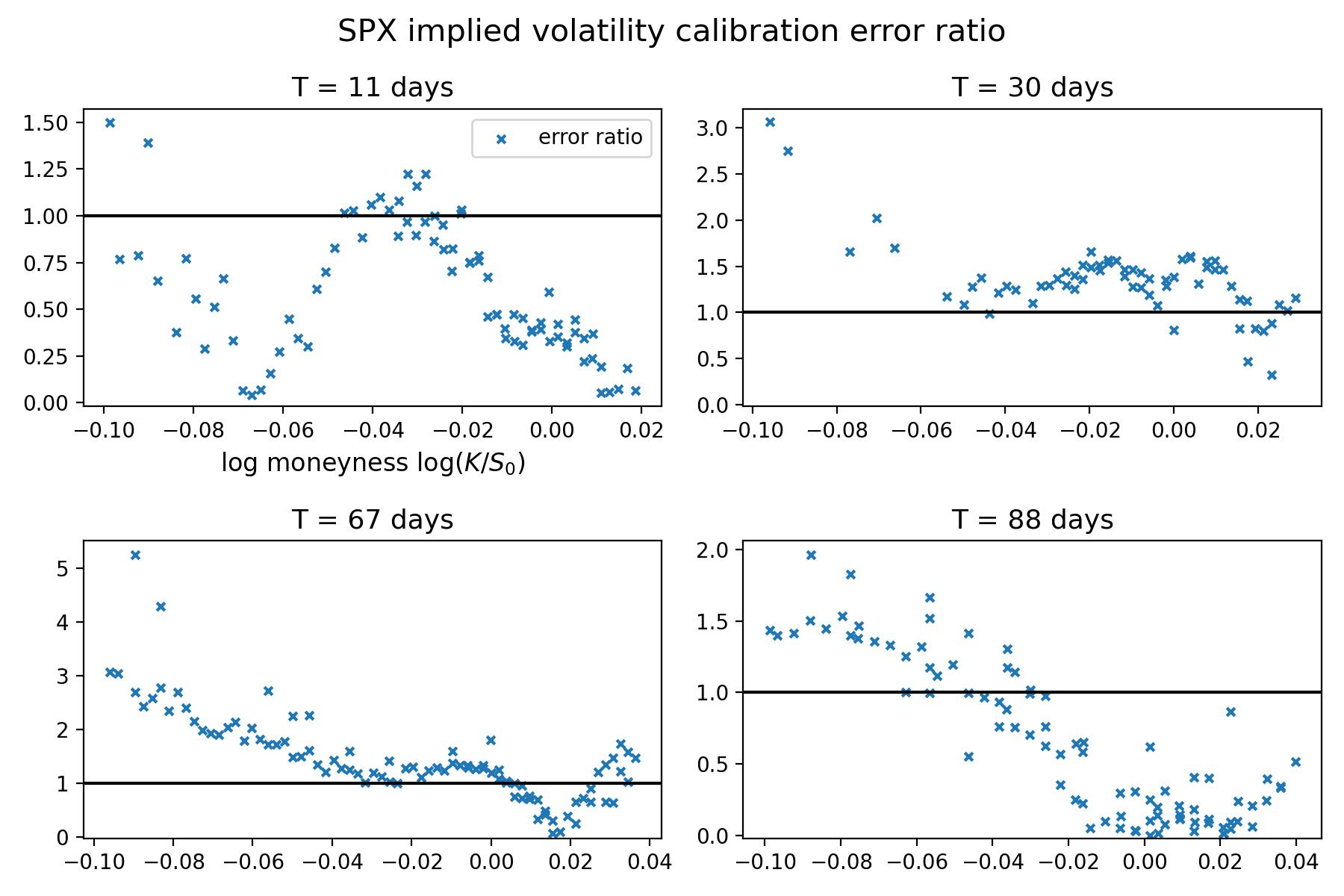}%
    \includegraphics[width=0.5\textwidth]{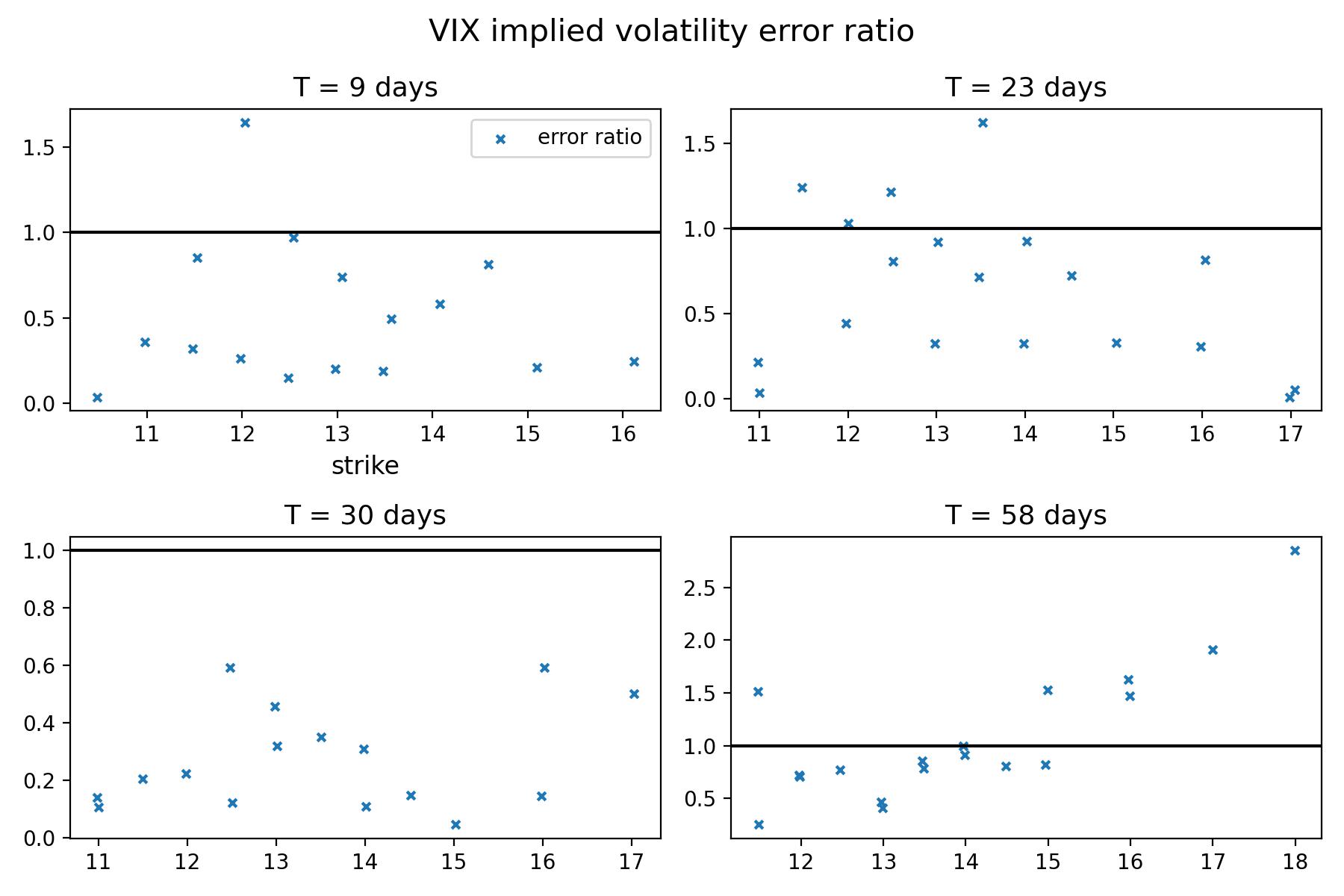}
    \caption{SPX--VIX implied volatility absolute calibration error as a multiplier of half of the bid-ask spread for 23 October 2017 using extracted forward variance curves.}
    \label{fig:joint_fit_error_ratio}
  \end{figure}

  \begin{figure}[H]
    \centering
    \includegraphics[width=0.6\textwidth]{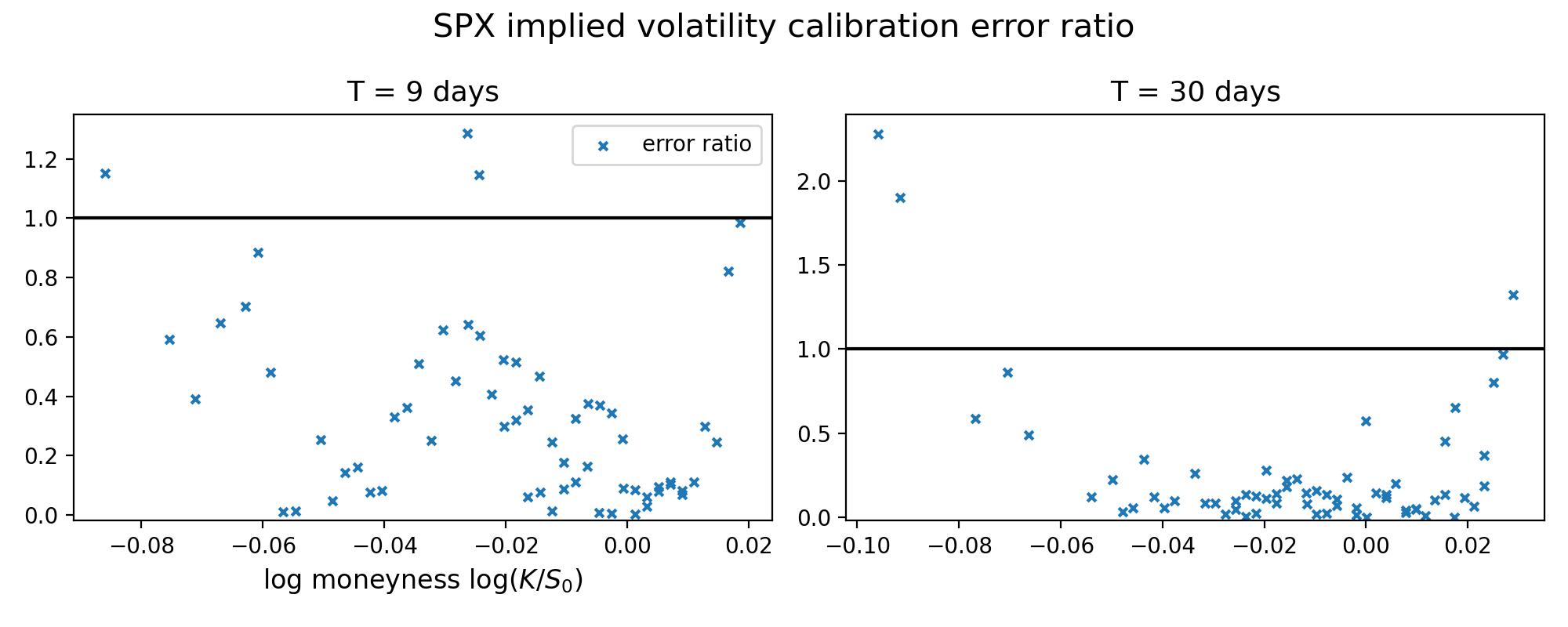}%
    \includegraphics[width=0.3\textwidth]{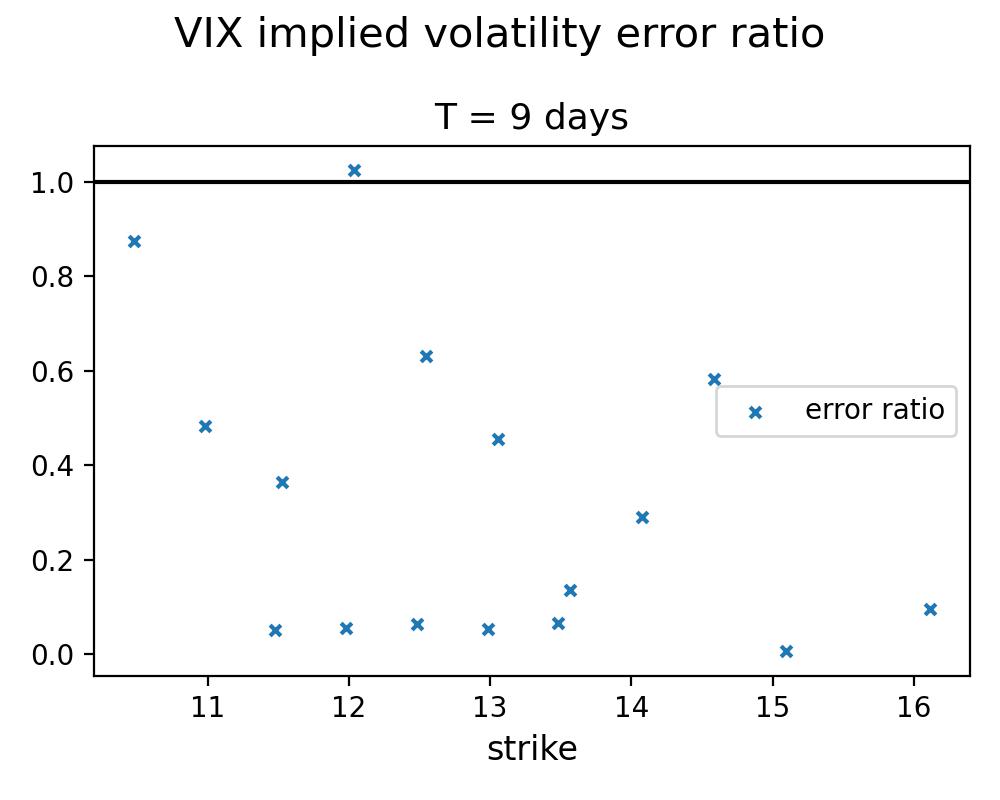}
    \caption{SPX--VIX implied volatility absolute calibration error as a multiplier of half of the bid-ask spread for 23 October 2017 using parametric forward variance curves.}
    \label{fig:joint_para1_error_ratio}
  \end{figure}

  \begin{figure}[H]
    \centering
    \includegraphics[width=0.6\textwidth]{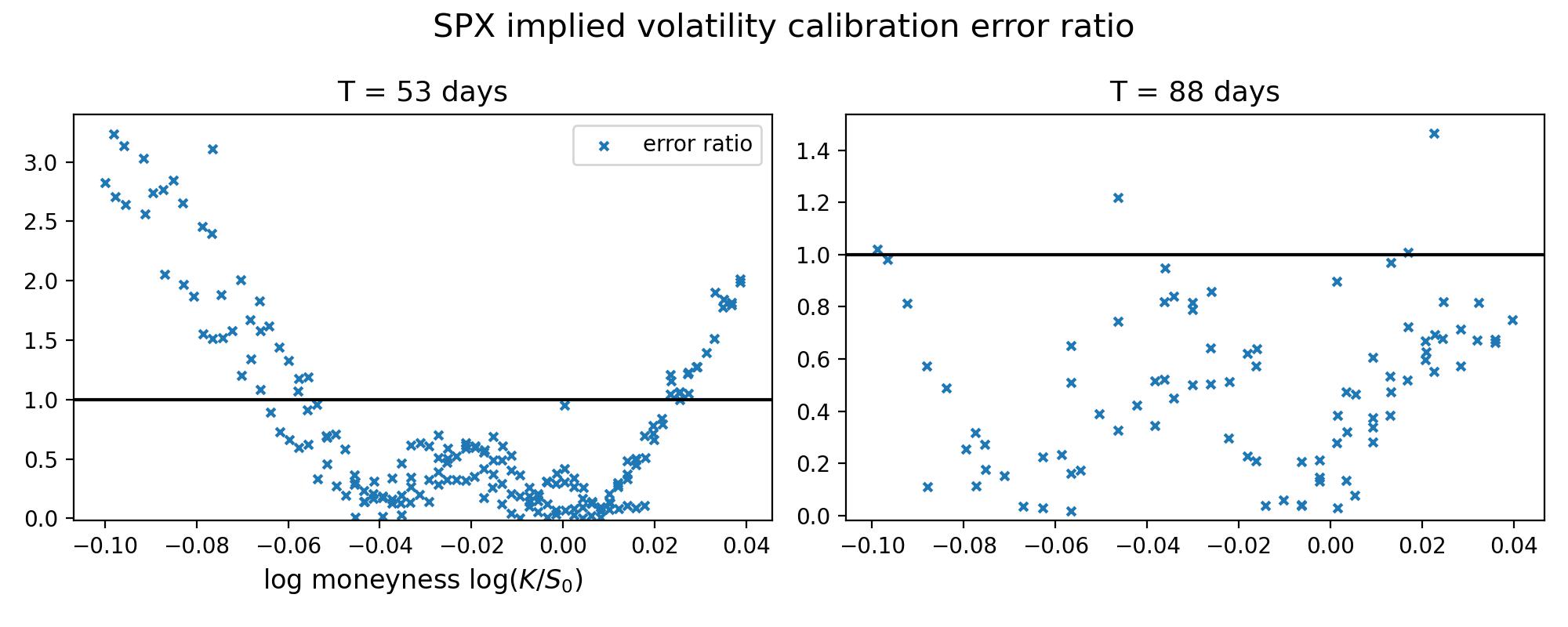}%
    \includegraphics[width=0.3\textwidth]{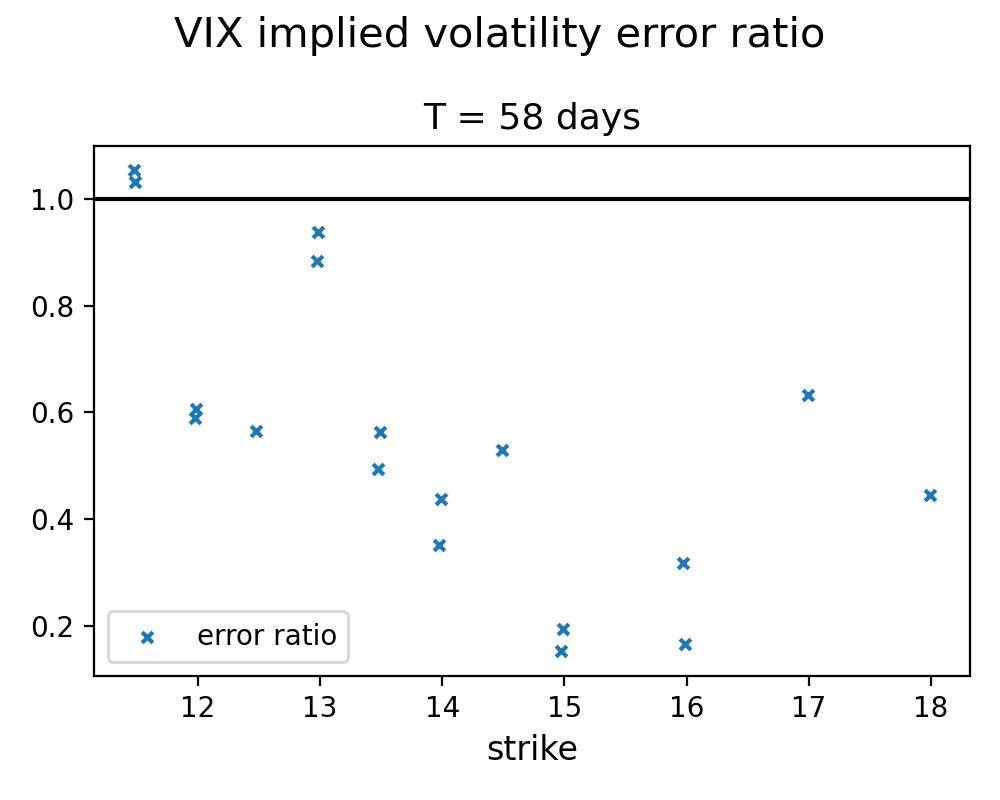}
    \caption{SPX--VIX implied volatility absolute calibration error as a multiplier of half of the bid-ask spread for 23 October 2017 using parametric forward variance curves.}
    \label{fig:joint_para1_error_ratio}
  \end{figure}
These graphs show that the absolute calibration error multiplier is largely below 1 (i.e.~within the bid-ask spread), especially around the at the money level for both SPX and VIX smiles.

{\appendix 
	\section{ 
On the martingale property of $S$}\label{A:mart}

We prove the true martingale property of the stock price $S$ in the quintic Ornstein-Uhlenbeck volatility model for constant forward variance curves. 

\begin{proposition}\label{prop:martingale}
Fix  $ \alpha_5 > 0$ and let $\xi_0$ be such that 
	 $ \xi_0(t) =  \xi^2 \E[p^2(X_t)]$,  for all  $ t\geq 0$ for some constant $\xi > 0$. If $\rho\leq 0$, 
 the process $S$ in \eqref{polynomial_model} is a true martingale.
\end{proposition}

A crucial ingredient for proving Proposition~\ref{prop:martingale}	is 
 the process
\begin{align}\label{S_rho}
    \Tilde S_t^{\rho} := \exp \left(-\frac{1}{2}\int_0^t b_{\rho}^2(X_s)ds +  \int_0^t b_{\rho}(X_s)dW_s\right),
\end{align}
with 
$$ b_{\rho}(x):= \rho \xi p(x) = \rho \xi (\alpha_0+\alpha_1 x + \alpha_3 x^3 + \alpha_5 x^5).$$ 

The martingality of the process $\Tilde S^\rho$ plays a crucial role in determining the martingale property of $S$ in \eqref{polynomial_model}. We first prove the martingality of $\Tilde S^{\rho}$.
	\begin{lemma}\label{propA} Under the assumptions of Proposition~\ref{prop:martingale}, the process $\Tilde S^{\rho}$ in \eqref{S_rho} is a true martingale. 
    \end{lemma} 
\begin{proof}
If $\rho=0$,  the process $\tilde S^{\rho}$ is (trivially) a martingale equal to $1$. For $\rho<0$, we make use of the general characterization in \cite[Theorem 2.1]{mijatovic2012martingale}.\footnote{ We are indebted to an anonymous referee for pointing out  this reference.}  Following the paper's notation, we set $l=-\infty, r = +\infty$ with $x \in (l,r)$ and write the process:
 \[
 dX_t = \mu (X_t) dt + \sigma(X_t)dW_t, \quad X_0 = 0, 
 \]
where $\mu (x) = ax, a\leq 0$ and $\sigma(x) = \eta$, with  $a= -(1/2-H)\varepsilon^{-1} $ and $\eta=\varepsilon^{H-1/2}$.   One can easily check that $\sigma (x) \neq 0$ for all $x \in \R$, $1/\mu$, $\mu/\sigma^2$ and $ b_\rho^2/\sigma^2$ are all locally integrable functions, so that the assumptions of \cite[Theorem 2.1]{mijatovic2012martingale} are met. Next, we introduce the auxiliary process $\tilde X$:
$$  d\Tilde X_t = (\mu + \eta b_{\rho})(\tilde X_t) dt + 
 \eta d\Tilde W_t,$$
with its corresponding scale function 
\begin{equation}\label{scale_function}
\Tilde s(x) = \int_c^x \Tilde p(y)dy,
\end{equation}
for  $c\in \R$. We also define the function $\Tilde p(y)$ as:
\[
\Tilde p(y) := \exp\left(\int_c^y -\frac{2au+2\eta b_\rho(u)}{\eta^2} du\right) = \exp\{f(y,c)\},
\]
with
\[
f(y,c) = -\frac{1}{\eta^2}\left[a(y^2-c^2)+2\eta (B_\rho(y)- B_\rho(c))\right],
\]
and $B_\rho$ the anti-derivative of $b_\rho$:
\[
B_\rho(y) = \rho \xi (\alpha_0 y+\frac{\alpha_1}{2}y^2 + \frac{\alpha_3}{4}y^4+\frac{\alpha_5}{6}y^6).
\]
The function $y\mapsto f(y,c)$ is a polynomial in $y$ of which the leading term $y^6$  has even power with positive coefficient since $ \alpha_5>0, \xi >0$ and $\rho < 0$. Therefore,
\begin{equation}
  \begin{aligned}
\Tilde s(+\infty) &= \int_c^{+\infty}\Tilde p(y) dy = \int_c^{+\infty} \exp\{f(y,c)\}dy = +\infty,\\
\Tilde s(-\infty) &= -\int_{-\infty}^{c} \Tilde p(y) dy = -\int_{-\infty}^c \exp\{f(y,c)\}dy = -\infty,
  \end{aligned}
\end{equation}
so that $\Tilde X_t$ does not exit the state space $(-\infty,+\infty)$ at the boundary $+\infty$ and $-\infty$. Applying \cite[Theorem 2.1]{mijatovic2012martingale}, which give us that $\Tilde S_t^\rho$ is a martingale. 
\end{proof}

\begin{proof}[Proof of Proposition~\ref{prop:martingale}]
It follows from \eqref{polynomial_model}, that   $S$ is a local martingale and non-negative, since 
\begin{align*}
	S_t = S_0 \exp\left( -\frac  12 \int_0^t \frac{\xi_0(u)}{\mathbb E[p^2(X_u)]} p^2(X_s) du  + \int_0^t\sqrt{  \frac{\xi_0(u)}{\mathbb E[p^2(X_u)]}} p(X_u) \left(\rho dW_u + \sqrt{1-\rho^2} dW_u^\perp\right)\right).
\end{align*}

It is therefore  a supermartingale by Fatou's lemma.  To show that is a true martingale, it suffices to argue that $\E[S_t]= S_0$ for any $t\in\R_+$. For this, we fix $t>0$ and we start by getting rid of $W^{\perp}$, by  conditionning on $\mathcal F^W_t$, to get 
\begin{align}
\E\left[S_t\right] &= S_0 \E\left[ \exp \{-\frac{1}{2}\int_0^t  \xi^2 p^2(X_s)ds + \rho \int_0^t   \xi p(X_s)dW_s\} \E\left[ \exp \{ \sqrt{1-\rho^2} \int_0^t  \xi^2 p(X_s)dW^{\perp}_s \} \Mid \mathcal{F}^{W}_t \right] \right].
\end{align}
Conditional on $\mathcal F^W_t$, the random variable $\int_0^t  \xi p(X_s)dW^{\perp}_s$ is a centered Gaussian random variable with variance $\int_0^t  \xi^2  p^2(X_s)ds$, which leads to
\begin{align}
\E\left[S_t\right]&=S_0 \E\left[ \tilde S^\rho_t\right],
\end{align}
with $\tilde S^{\rho}$ defined in \eqref{S_rho}.
From Lemma \ref{propA}, $\tilde S^\rho$ is a martingale, which shows that $\E\left[S_t\right]=S_0$ and completes the proof.
\end{proof}

}

\bibliographystyle{plainnat}
\bibliography{bibl.bib}

\end{document}